\documentclass[12pt,draftcls,onecolumn]{IEEEtran}
\usepackage{amsfonts}
\usepackage{psfrag}
\usepackage{amsmath}
\usepackage{subfigure}
\usepackage{bm}
\usepackage{cite}
\usepackage{graphicx}          
  \usepackage{amssymb}

\newtheorem{lem}{Lemma}
\newtheorem{thm}{Theorem}
\newtheorem{rem}{Remark}
\newtheorem{Def}{Definition}

\newtheorem{Assumption}{Assumption}

\begin{document}
%
\title{Reaching a Consensus in Networks of High-Order Integral Agents under Switching Directed Topology
\thanks{The authors are with the State Key Laboratory of Intelligent Control and Management of Complex Systems, Institute of Automation, Chinese
Academy of Sciences, Beijing 100190, China.}\thanks{Please address all correspondences to Dr. Long Cheng: chenglong@compsys.ia.ac.cn, Phone: +86-10-62568112, Fax:
+86-10-82629972.}}

\author{Long~Cheng,~Zeng-Guang Hou,~and~Min Tan}

\maketitle

\vspace{0.5cm}

\begin{abstract}
Consensus problem of high-order integral multi-agent systems under switching directed topology is considered in this study. Depending on whether the agent's full state is available or not, two distributed protocols are proposed to ensure that states of all agents can be convergent to a same stationary value. In the proposed protocols, the gain vector associated with the agent's (estimated) state and the gain vector associated with the relative (estimated) states between agents are designed in a sophisticated way. By this particular design, the high-order integral multi-agent system can be transformed into a first-order integral multi-agent system. And the convergence of the transformed first-order integral agent's state indicates the convergence of the original high-order integral agent's state if and only if all roots of the polynomial, whose coefficients are the entries of the gain vector associated with the relative (estimated) states between agents, are in the open left-half complex plane. Therefore, many analysis techniques in the first-order integral multi-agent system can be directly borrowed to solve the problems in the high-order integral multi-agent system. Due to this property, it is proved that to reach a consensus, the switching directed topology of multi-agent system is only required to be ``uniformly jointly quasi-strongly connected'', which seems the mildest connectivity condition in the literature. In addition, the consensus problem of discrete-time high-order integral multi-agent systems is studied. The corresponding consensus protocol and performance analysis are presented. Finally, three simulation examples are provided to show the effectiveness of the proposed approach.
\end{abstract}

\begin{IEEEkeywords}
Multi-agent systems, high-order integral agent, linear dynamics, switching directed topology, consensus.
\end{IEEEkeywords}

\IEEEpeerreviewmaketitle


\section{Introduction}
Recent years have witnessed a rapid development on the consensus or state agreement of multi-agent systems due to its central
role in distributed coordination tasks. Roughly speaking, the consensus problem means that the states of all agents are convergent to a same value in a distributed manner. In control community, considerable research efforts have been made to address this problem from different aspects such as the dynamics of agent \cite{Zhao13TSMCB, Cheng09TSMCB}, connectivity of communication topology \cite{Chen13TSMCB}, and communication constraints \cite{Zhong12TSMCB, Meng11TSMCB}. It is noted that for the convenience of studying the last two aspects, most results assume that the agent is modeled by the first-order/second-order integral dynamics. However, due to the diversity of real-world agents, these simple dynamics are insufficient to fully describe the agent's dynamical behavior. To fill this gap, an intuitive way is to investigate the agent described by the high-order integral/general linear time-invariant dynamics first.

In the literature, consensus problems of high-order integral/linear multi-agent systems have been addressed in \cite{Xiao07IETCTA,Ma10TAC,Cheng11Automatica,Kim11TAC,Zhang12Automatica,Zhang11TAC,Li10TCASI,Yu11TCASI,Cheng12TAC,Tian12Automatica,You11TAC,
Wang08AJC,Jiang10IJC,Ni10SCL,Xu13IET,Su12TAC,Su12Automatica,Su12TSMCB}, to name a few. Some early results are presented in \cite{Xiao07IETCTA} where a state-feedback based consensus protocol is proposed for linear time-invariant multi-agent systems under fixed communication topology. Since the full state of linear agents may be unavailable in some scenarios, several output-feedback based consensus protocols have been proposed as well, for example, the static output-feedback based protocol \cite{Ma10TAC} and the dynamic output-feedback based protocols \cite{Cheng11Automatica}. And the consensusability of linear multi-agent systems under these output-feedback protocols is analyzed. In \cite{Kim11TAC,Zhang12Automatica}, uncertainties in the dynamics of high-order integral/linear agents are considered, and the consensus can still be achieved by employing the internal model approach and the neural-network-based adaptive approach, respectively. An interesting leader-following consensus protocol is proposed for linear multi-agent systems in \cite{Zhang11TAC}. Furthermore, in \cite{Li10TCASI, Yu11TCASI}, the relationship between the consensus of high-order integral/linear multi-agent systems and the synchronization in complex networks is discussed, and the concept of consensus region is introduced. In addition, there are also some papers considering the consensus of high-order integral/linear multi-agent systems with communication constraints such as communication noises \cite{Cheng12TAC}, communication delays \cite{Tian12Automatica}, and quantization effect \cite{You11TAC}. It is noted that all above results assume that the multi-agent system works under fixed communication topology. However, in some applications of mobile agents, the creation and loss of communication links frequently occur due to the sensor's limited working region. Therefore, the study on the consensus of high-order integral/linear multi-agent systems under switching topology is of more practical use.

Some attempts in this direction have been made in \cite{Wang08AJC,Jiang10IJC,Ni10SCL,Xu13IET,Su12TAC,Su12Automatica,Su12TSMCB}. One common assumption of these papers is that the topology needs to be undirected,
which requires the bidirectional information exchange between agents.
Furthermore, in \cite{Wang08AJC}, the communication topology is required to be ``frequently connected'' rather than ``jointly connected''.
In \cite{Jiang10IJC,Ni10SCL,Xu13IET}, to design the control gain, the eigenvalues of graph Laplacian matrix (a certain global information) should be known.
In \cite{Su12TAC,Su12Automatica}, the result is highly dependent on the doubly stochastic property of Laplacian matrix of undirected graph. In \cite{Su12TSMCB}, the consensus tracking problem of linear multi-agent systems under switching topologies is investigated, however, all agents need to know the system matric in the so-called exosystem which generates the  tracking reference signal. 
It is also noted that consensus of first-order integral multi-agent systems with switching directed topology has been extensively studied.
The mildest connectivity condition on the switching directed topology seems ``uniformly jointly quasi-strongly connected'' \cite{Ren05TAC,Shi12arXiv}.
Then a question arises consequently: is it possible to obtain the counterpart result in the high-order integral/linear multi-agent systems?

This paper tries to give a positive answer to the above question. Two novel protocols, the state-feedback based protocol and the output-feedback based protocol, are proposed to solve the consensus problem of continuous-time high-order integral multi-agent systems. By the proposed protocol, it is interesting to find that the original high-order integral multi-agent system can be transformed into a first-order integral multi-agent system. And the convergence of the transformed first-order integral agent's state implies the convergence of the high-order integral agent's state if and only if all roots of the polynomial, whose coefficients are the entries of control gain vector of the proposed protocol, are in the open left-half complex plane. Therefore, most results in the first-order integral multi-agent system can be generalized to the high-order integral multi-agent system directly. Then it is proved that the connectivity condition on the consensus of high-order integral multi-agents under switching directed topology is that the switching topology is uniformly jointly quasi-strongly connected. To the best of the authors' knowledge, it seems the mildest connectivity condition in the literature of high-order integral multi-agent systems. In addition, results obtained in the continuous-time domain can be extended to the discrete-time domain as well. Finally, the theoretical analysis is validated by three illustrative examples.

The reminder of this paper is organized as follows. Section {\ref{Problem}} introduces the problem formulation and some
preliminary results. Section {\ref{continuous}} presents the consensus protocol for high-order multi-agent systems in the continuous-time domain and gives the corresponding performance analysis. Section {\ref{discrete}} extends the results in the continuous-time domain to the discrete-time domain. Illustrative examples  are provided in Section {\ref{Simulation}}. Section {\ref{Conclusion}} concludes this paper
with final remarks.

The following notations are used throughout this paper: $\mathbb{R}$ denotes the set of real numbers; $\mathbb{N}$ denotes the set of natural number; $1_n = (1,\cdots,1)^T \in
\mathbb{R}^n$; $0_n = (0,\cdots,0)^T \in \mathbb{R}^n$; $\mathbb{C}^k_n = \frac{n!}{k!(n-k)!}$; $I_n$ denotes the $n\times n$ dimensional
identity matrix; $\otimes$ denotes the Kronecker operator;
For a given matrix $X$, $X^T$ denotes its transpose; $\|X\|_2$ denotes its Euclidean norm; $\|X\|_F$ denotes its Frobenius norm.


\section{Problem Formation and Preliminaries}\label{Problem}
In the literature, the graph theory is commonly employed to describe the communication among agents.
Let $\mathcal{G} = \{\mathcal{V}_\mathcal{G}, \mathcal{E}_\mathcal{G}\}$ be a digraph where $\mathcal{V}_\mathcal{G} = \{v_1,\cdots,v_N\}$ is the node
set and $\mathcal{E}_\mathcal{G} \subseteq \mathcal{V}_\mathcal{G} \times \mathcal{V}_\mathcal{G}$ is the edge set. The $i$th agent in the network is
represented by the node $v_i$. The edge $e_{ij}$ starting from $v_j$ to $v_i$ belongs to $\mathcal{E}_\mathcal{G}$ if and only if the $i$th agent can
receive information from the $j$th agent. In this paper, it is assumed that there is no self-edge in the graph,
\emph{i.e.}, $e_{ii} \notin \mathcal{E}_\mathcal{G}, i=1,\cdots,N$. The neighbor set of the $i$th agent is defined
as $\mathcal{N}_i = \{v_j \in \mathcal{V}_\mathcal{G} | e_{ij} \in \mathcal{E}_\mathcal{G}\}$. A directed path in $\mathcal{G}$ is a sequence of
distinct nodes $v_{i_1},v_{i_2},\cdots,v_{i_m}$ such that $e_{i_{s+1}i_s} \in \mathcal{E}_\mathcal{G}, s=1,\cdots,m-1$. The directed graph $\mathcal{G}$ is
called strongly connected if for any two distinct nodes $v_i$ and $v_j$ there is a directed path starting from $v_j$ to $v_i$. A node is called a center/root
if there are directed paths starting from this node to any other nodes in $\mathcal{V}_\mathcal{G}$. The directed graph $\mathcal{G}$ is called
quasi-strongly connected if $\mathcal{G}$ has at least one center node.  


Due to the link failure and packet loss, the communication channel of multi-agent systems is usually time-variant. Therefore, it is assumed in this paper that the communication topology is modeled by a set of graphs $\mathcal{G}_{\sigma(t)} = \{\mathcal{V}_\mathcal{G}, \mathcal{E}_{\mathcal{G}_{\sigma(t)}}\}$. $\sigma(\cdot): [0,\infty) \rightarrow \mathcal{S}$ is a piecewise constant function whose value at time $t$ is the index of the graph representing the agent's communication topology at time $t$. $\mathcal{S}$ denotes the index set of all possible graphs and $\mathcal{S}$ has finite elements because the node set $\mathcal{V}_\mathcal{G}$ is a finite set. The following assumption is applied to the switching signal $\sigma(\cdot)$, which means that the communication topology does not change too fast.

\begin{Assumption}\label{assumption1}
Let $(t_0,t_1,\cdots)$ be the sequence of time points at which the piecewise constant function $\sigma(t)$ switches. The dwell times $\tau_i = t_{i+1} - t_i$ $(i=0,1,\cdots)$ have a uniform lower bound $\tau_D > 0$.
\end{Assumption}

The union graph of $\mathcal{G}_{\sigma(t)}$ over the time interval $[t_1,t_2)$ is defined as $\mathcal{G}_{[t_1,t_2)} = (\mathcal{V}_\mathcal{G}, \bigcup_{t\in[t_1,t_2)}\mathcal{E}_{\mathcal{G}_\sigma(t)})$. One concept on the connectivity of the union graph is given in the following.
\begin{Def}
$\mathcal{G}_{\sigma(t)}$ is said to be uniformly jointly quasi-strongly connected if there exists a constant $T>0$ such that $\mathcal{G}_{[t,t+T)}$ is quasi-strongly connected for any $t>0$.
\end{Def}

The above section introduces the communication topology of the network of $N$ agents. The other essential part of modeling the multi-agent system is the agent's dynamics. In this paper, the dynamical behavior of each agent is described by the following $m$-dimensional high-order integral dynamics
\begin{IEEEeqnarray}{ccc}\label{equ:agentdynamics}
\begin{cases}
&\dot{x}_{i,1}(t)=x_{i,2}(t),\\
&\quad\quad\vdots\\
& \dot{x}_{i,m-1}(t)=x_{i,m}(t), \\
&\dot{x}_{i,m}(t)=u_{i}(t),
\end{cases}
\end{IEEEeqnarray}
which can be written in the compact form
\begin{IEEEeqnarray}{lll}\label{equ:agentdynamicscompact}
&\dot{x}_i(t) = Ax_i(t) + Bu_i(t),\IEEEnonumber\\
& A = \begin{pmatrix} 0_{m-1} & I_{m-1}\\ 0 & 0^T_{m-1}\end{pmatrix}, \IEEEnonumber\\
& B= (0,\cdots,0,1)^T, x_i(t) = (x_{i,1}(t),\cdots,x_{i,m}(t))^T.
\end{IEEEeqnarray}

The control objective is to solve the distributed consensus problem of this high-order integral multi-agent system which is defined as follows.
\begin{Def}\label{def1}
The multi-agent system is said to reach a consensus if under certain protocol $u_i(t)$, there exists a vector $x^*\in \mathbb{R}^m$ such that
\begin{displaymath}
\lim_{t\to\infty}x_i(t) = x^*,\quad i=1,\cdots,N.
\end{displaymath}
If the consensus protocol $u_i(t)$ only employs the information from the neighbor agents $j \in \mathcal{N}_i$, then this protocol is called a distributed consensus protocol.
\end{Def}

Before closing this section, the following result on the robust consensus of first-order integral multi-agents is provided \cite{Shi12arXiv}.
\begin{lem}\label{keylem}
For the first-order integral multi-agent system with external disturbances
\begin{equation}\label{equ:firstorder}
\dot{r}_i(t) = \sum_{j\in \mathcal{N}_i(\sigma(t))}\alpha^{ij}_{\sigma(t)}(r_j(t) - r_i(t)) + \omega_i(t), \quad i=1,\cdots,N,
\end{equation}
where $r_i(t) \in \mathbb{R}$ denotes the state of the $i$th first-order integral agent, $\mathcal{N}_i(\sigma(t))$ denotes the neighbor set of agent $i$ at time $t$, $\alpha^{ij}_{\sigma(t)} > 0$ is the weight constant associated with the edge $e_{ij}$ in graph $\mathcal{G}_{\sigma(t)}$, and $\omega_i(t) \in \mathbb{R}$ is the continuous disturbance signal. For $\forall r_i(0)$ and $\forall \omega_i(t)$ satisfying $\sup_{t\in[0,\infty)}|\omega_i(t)| < \infty$ and $\lim_{t\to\infty}\omega_i(t) = 0$, then $\lim_{t\to\infty} (r_i(t)-r_j(t)) = 0$ if $\mathcal{G}_{\sigma(t)}$ is uniformly jointly quasi-strongly connected;
\end{lem}
\begin{proof}
 See the proof of Proposition 1 in \cite{Shi12arXiv}.
 \end{proof}

\section{Consensus Protocol and Related Analysis}\label{continuous}
\subsection{State-Feedback Based Consensus Protocol}
In this subsection, it is assumed that the full state information of each agent is available for designing the consensus protocol. Inspired by the results presented in \cite{Cheng12TAC}, the following consensus protocol is proposed
\begin{equation}\label{equ:protocol}
u_i(t) = K_1x_i(t) - \sum_{j\in \mathcal{N}_i(\sigma(t))}\alpha^{ij}_{\sigma(t)}K_2(x_i(t)-x_j(t)),
\end{equation}
where $K_1 = (0,-a_1,-a_2,\cdots,-a_{m-1}) \in \mathbb{R}^{1\times m}$ and $K_2 = (a_1,a_2,\cdots,a_{m-1},1) \in \mathbb{R}^{1\times m}$. The parameters $(a_1,\cdots,a_{m-1})$ are constant control gains to be designed later. $\alpha^{ij}_{\sigma(t)}$ is defined in (\ref{equ:firstorder}).

Substituting (\ref{equ:protocol}) into (\ref{equ:agentdynamics}) obtains that
\begin{equation}\label{equ:tmp4}
\dot{X}(t) = (I_N \otimes (A+BK_1) - L_{\mathcal{G}_{\sigma(t)}}\otimes (BK_2))X(t),
\end{equation}
where $X(t) = (x^T_1(t),\cdots,x^T_N(t))^T$, the off-diagonal entry $(i,j)$ of $L_{\mathcal{G}_{\sigma(t)}}$ is $-\alpha^{ij}_{\sigma(t)}$, and the diagonal entry $(i,i)$ of $L_{\mathcal{G}_{\sigma(t)}}$ is $\sum_{j\in \mathcal{N}_i(\sigma(t))}\alpha^{ij}_{\sigma(t)}$. Here $L_{\mathcal{G}_{\sigma(t)}}$ is called the Laplacian matrix of the graph $\mathcal{G}_{\sigma(t)}$.


Let $\bar{X}(t)  = (I_N\otimes K_2)X(t) \equiv (\bar{x}_1(t),\cdots,\bar{x}_N(t))^T \in \mathbb{R}^N$, then multiplying $(I_N\otimes K_2)$ at the both sides of (\ref{equ:tmp4}) obtains that
\begin{IEEEeqnarray}{lll}\label{equ:tmp5}
\dot{\bar{X}}(t) &= (I_N \otimes K_2(A+BK_1) - L_{\mathcal{G}_{\sigma(t)}}\otimes (K_2BK_2))X(t) \IEEEnonumber\\
&=  -(L_{\mathcal{G}_{\sigma(t)}}\otimes (K_2))X(t) \IEEEnonumber\\
&= -(L_{\mathcal{G}_{\sigma(t)}}\otimes 1)(I_N\otimes K_2)X(t) = -L_{\mathcal{G}_{\sigma(t)}}\bar{X}(t).
\end{IEEEeqnarray}

From (\ref{equ:tmp5}), it is interesting to find that under the proposed protocol, the original high-order integral multi-agent system has been transformed into a first-order integral multi-agent system. Before proceeding with proving that $\bar{x}_i(t) = K_2x_i(t)$ $(i=1,\cdots,N)$ can reach a consensus, we should first study whether the convergence of $\bar{x}_i(t)$ implies the convergence of $x_i(t)$, which is answered by the following lemma.

\begin{lem}\label{lem1}
Consider the following non-homogeneous linear differential equation
\begin{equation}\label{equ:lemdifeqn}
r^{(m-1)}(t) + a_{m-1}r^{(m-2)}(t) + \cdots + a_2r^{(1)}(t) + a_1r(t) = f(t),
\end{equation}
where $r(t) \in \mathbb{R}$ and $f(t) \in \mathbb{R}$ is a continuous function satisfying $\lim_{t\to\infty}f(t) = f^*$. Let the characteristic equation associated with (\ref{equ:lemdifeqn}) be
\begin{equation}\label{equ:characteristicequation}
s^{m-1} + a_{m-1}s^{m-2} + \cdots a_2s + a_1 = 0.
\end{equation}
Then,
for any initial state $(r(0),\cdots,r^{(m-2)}(0))$, $(r(t),\cdots,r^{(m-2)}(t))$ are convergent if and only if all roots of (\ref{equ:characteristicequation}) are in the open left-half complex plane. In addition, $\lim\limits_{t\to\infty}(r(t),r^{(1)}(t),\cdots,r^{(m-2)}(t)) = (f^*/a_1,0,\cdots,0)$.
\end{lem}
\begin{proof}
See the proof in the Appendix.
\end{proof}

Then the following main result can be obtained by Lemmas \ref{keylem} and \ref{lem1}.
\begin{thm}\label{thm:main}
The proposed protocol defined by (\ref{equ:protocol}) can solve the consensus problem of high-order integral multi-agent systems if all roots of (\ref{equ:characteristicequation}) are in the open left-half complex plane and the communication topology $\mathcal{G}_{\sigma(t)}$ is uniformly jointly quasi-strongly connected.
\end{thm}
\begin{proof}
First, if $\mathcal{G}_{\sigma(t)}$ is uniformly jointly quasi-strongly connected, by (\ref{equ:tmp5}) and Lemma \ref{keylem}, it follows that for $\forall \bar{x}_i(0), \bar{x}_j(0)$, $\lim_{t\to\infty}(\bar{x}_i(t)-\bar{x}_j(t)) = 0$, because $\omega_i(t) = 0$ satisfies the conditions required in Lemma \ref{keylem}.

Next, it is proved that there exists $\bar{x}^*$ such that $\lim_{t\to\infty}\bar{x}_i(t) = \bar{x}^*$, $i=1,\cdots,N$ if and only if  $\lim_{t\to\infty}(\bar{x}_i(t)-\bar{x}_j(t)) = 0$, $\forall i,j=1,\cdots,N$.
The sufficiency is obvious. To prove the necessity, let $\bar{x}_{\max}(t) = \max\{\bar{x}_1(t),\cdots,\bar{x}_N(t)\}$ and $\bar{x}_{\min}(t) = \min\{\bar{x}_1(t),\cdots,\bar{x}_N(t)\}$. By (\ref{equ:tmp5}), it follows that $\dot{\bar{x}}_{\max}(t) \leq 0$ and $\dot{\bar{x}}_{\min}(t) \geq 0$, $\forall t\geq 0$. This together with $\lim_{t\to\infty}(\bar{x}_i(t)-\bar{x}_j(t)) = 0$ lead to that there exits $\bar{x}^*$ such that $\lim_{t\to\infty}\bar{x}_i(t) = \bar{x}^*$, $i=1,\cdots,N$.

If all roots of (\ref{equ:characteristicequation}) are in the open left-half complex plane, it can be obtained by Lemma \ref{lem1} that $\lim_{t\to\infty}x_i(t) = x^*$, $i=1,\cdots,N$ where $x^* = (\bar{x}^*/a_1,0,\cdots,0)^T$, which means that the proposed protocol can solve the consensus problem. In addition, ``all roots of (\ref{equ:characteristicequation}) are in the open left-half complex plane'' is also necessary for solving the consensus problem by the necessity proof of Lemma \ref{lem1}.
\end{proof}

\begin{rem}
The most interesting feature of the consensus protocol defined by (\ref{equ:protocol}) is that under the proposed protocol, the consensus problem of high-order integral multi-agent systems is equivalent with the consensus problem of first-order integral multi-agent systems. Therefore, many existing results (for example, the consensus problem with delayed communication \cite{Saber04TAC}) in the literature can be generalized to the high-order integral case directly.
\end{rem}

\begin{rem}
The group decision value $x^*$ is determined by two factors: the agents' initial states and the communication topology $\mathcal{G}_{\sigma(t)}$.
In the switching topology case, it is usually hard to give an explicit solution to $x^*$. However, if $\mathcal{G}_{\sigma(t)}$ is balanced at any time
(the graph $\mathcal{G}_{\sigma(t)}$ is called balanced if $1^TL_{\mathcal{G}_{\sigma(t)}} = 0^T_N$), then by \eqref{equ:tmp5}, $1^T_N\dot{\bar{X}}(t) = - 1^T_NL_{\mathcal{G}_{\sigma(t)}}\bar{X}(t) = 0^T_N$.
Hence $1^T_N\bar{X}(t) = 1^T_N\bar{X}(0)$. This together with $\lim_{t\to\infty}\bar{x}_i(t) = \bar{x}^*$, $i=1,\cdots,N$
leads to $\bar{x}^* = \frac{1}{N}\sum^N_{i=1}\bar{x}_i(0) = \frac{1}{N}\sum^N_{i=1}K_2{x}_i(0)$. Then $x^* = (\frac{1}{a_1N}\sum^N_{i=1}K_2{x}_i(0),0,\cdots,0)^T$.
\end{rem}

\subsection{Output-Feedback Based Consensus Algorithm}
In this subsection, it is assumed that the agent's full state is not available any more. Instead, only the agent's output can be used for the consensus protocol design, which is modeled by the following equation
\begin{equation}\label{equ:output}
y_i(t) = Cx_i(t),
\end{equation}
where $C = (c_1,\cdots,c_m)$. Here it is assumed that the pair $(A,C)$ is detectable.

Motivated by the dynamic output feedback consensus protocol proposed in \cite{Cheng11Automatica}, the following protocol is proposed
\begin{IEEEeqnarray}{ll}\label{equ:observerprotocol}
& u_i(t) = K_1s_i(t) - \sum_{j\in \mathcal{N}_i(\sigma(t))}\alpha^{ij}_{\sigma(t)}K_2(s_i(t) - s_j(t)), \IEEEyessubnumber \label{equ:observerprotocol1}\\
& \dot{s}_i(t) = (A+K_3C)s_i(t) + Bu_i(t) - K_3y_i(t), \IEEEyessubnumber \label{equ:observerprotocol2}
\end{IEEEeqnarray}
where $K_1$ and $K_2$ are defined in (\ref{equ:protocol}), $K_3 \in \mathbb{R}^{m\times 1}$ is designed in such a way that $A+K_3C$ is \emph{Hurwitz}. The basic idea behind this protocol is that design an observer (\ref{equ:observerprotocol2}) to dynamically estimate the agent's full state, and then use the estimated state to replace the actual state in (\ref{equ:protocol}), which results in (\ref{equ:observerprotocol1}).

By (\ref{equ:agentdynamicscompact}) and (\ref{equ:observerprotocol2}), it follows that
\begin{displaymath}
d(x_i(t) - s_i(t))/dt = (A+K_3C)(x_i(t) - s_i(t)).
\end{displaymath}
Then
\begin{equation}\label{equ:xssolution}
x_i(t) = s_i(t) + \exp((A+K_3C)t)(x_i(0)-s_i(0)).
\end{equation}

Since $A+K_3C$ is \emph{Hurwitz}, it can be proved that $\lim_{t\to\infty}(x_i(t)-s_i(t)) = 0$. Therefore, to prove that $x_i(t)$ $(i=1,\cdots,N)$ can reach a consensus is equivalent to prove that $s_i(t)$ $(i=1,\cdots,N)$ can reach a consensus. Substituting (\ref{equ:xssolution}) into (\ref{equ:observerprotocol2}) obtains that
\begin{multline}\label{equ:observerdynamics}
\dot{S}(t) = (I_N\otimes(A+BK_1) - L_{\mathcal{G}_{\sigma(t)}}\otimes BK_2)S(t)\\ - (I_N\otimes(K_3C\exp((A+K_3C)t)))(X(0)-S(0)),
\end{multline}
where $S(t) = (s^T_1(t),\cdots,s^T_N(t))^T$ and $X(t)$ is defined in (\ref{equ:tmp4}).

Similarly, let $\bar{S}(t)  = (I_N\otimes K_2)S(t) \equiv (\bar{s}_1(t),\cdots,\bar{s}_N(t))^T \in \mathbb{R}^N$, then multiplying $(I_N\otimes K_2)$ at the both sides of (\ref{equ:observerdynamics}) obtains that
\begin{multline}\label{equ:observerdynamics}
\dot{\bar{S}}(t) =  - (I_N\otimes(K_2K_3C\exp((A+K_3C)t)))(X(0)-S(0)) \\-L_{\mathcal{G}_{\sigma(t)}}\bar{S}(t) \equiv -L_{\mathcal{G}_{\sigma(t)}}\bar{S}(t) + \bar{\omega}(t),
\end{multline}
where $\bar{\omega}(t) = (\bar{\omega}_1(t),\cdots,\bar{\omega}_N(t))^T$ and $\bar{\omega}_i(t) = (K_2K_3C\exp((A+K_3C)t))(x_i(0)-s_i(0))$.

From (\ref{equ:observerdynamics}), the original consensus problem has been transformed into the robust consensus problem of first-order integral multi-agent systems. Similar with the method employed in Theorem \ref{thm:main}, the following result shows that the output-feedback based consensus protocol defined by (\ref{equ:observerprotocol}) can also solve the consensus problem under the same conditions of Theorem \ref{thm:main}.

\begin{thm}\label{thm:main2}
Assume that $A+K_3C$ is \emph{Hurwitz}. The proposed protocol defined by (\ref{equ:observerprotocol}) can solve the consensus problem of high-order integral multi-agent systems if all roots of (\ref{equ:characteristicequation}) are in the open left-half complex plane and the communication topology $\mathcal{G}_{\sigma(t)}$ is uniformly jointly quasi-strongly connected.
\end{thm}
\begin{proof} First, let $\Phi_{\bar{S}}(t,s)$ be the state transition matrix of (\ref{equ:observerdynamics}). Due to Assumption \ref{assumption1}, $\Phi_{\bar{S}}(t,s)$ can be written as
\begin{multline}
\Phi_{\bar{S}}(t,s) = \exp(-L_{\mathcal{G}_{\sigma(t_p)}}(t-t_p))\exp(-L_{\mathcal{G}_{\sigma(t_{p-1})}}(t_p-t_{p-1}))\\\times\cdots\times\exp(-L_{\mathcal{G}_{\sigma(t_{q})}}(t_{q+1}-s)), \; t\geq s,
\end{multline}
where $t_p$ is the largest time switching point smaller than $t$ and $t_q$ is the largest time switching point smaller than $s$. In \cite{Ren04Chapter}, it has been proved that for any $\Delta \geq 0$, $\exp(-L_{\mathcal{G}_{\sigma(t_i)}}\Delta)$ is a stochastic matrix. Then $\Phi_{\bar{S}}(t,s)$ is also a stochastic matrix. Therefore, $\|\Phi_{\bar{S}}(t,s)\|_F\leq N^2$, $\forall t \geq s$.

If $\mathcal{G}_{\sigma(t)}$ is uniformly jointly quasi-strongly connected, by the same analysis in the proof of Theorem \ref{thm:main}, the solution to the homogeneous differential equation associated with (\ref{equ:observerdynamics}), $\dot{\bar{S}}(t) =  -L_{\mathcal{G}_{\sigma(t)}}\bar{S}(t)$, is convergent under any initial states. Therefore, there exists a matrix $\Phi^*_{\bar{S}}$ such that $\lim_{t\to\infty}\Phi_{\bar{S}}(t,0) = \Phi^*_{\bar{S}}$.

The solution to  (\ref{equ:observerdynamics}) can be written as
\begin{equation}\label{equ:observersolution}
\bar{S}(t) = \Phi_{\bar{S}}(t,0)\bar{S}(0) + \int^t_0\Phi_{\bar{S}}(t,s)\bar{\omega}(s)ds.
\end{equation}

Next, let $\mu_1,\cdots,\mu_m$ be the eigenvalues of $A+K_3C$. Then $\bar{\omega}_k(t)$ in (\ref{equ:observerdynamics}) is the linear combination of $t^je^{\mu_i t}$, $i,j=1,\cdots,m$. Since the real parts of $\mu_1,\cdots,\mu_m$ are all in the open left-half complex plane, it can be proved that $\lim_{t\to\infty}t^je^{\mu_it} = 0$ and $\int^{\infty}_0t^je^{\mu_it}dt < \infty$. Then $\sup_{t\in[0,\infty)}|\omega_k(t)| < \infty$, $\lim_{t\to\infty}\bar{\omega}_k(t) = 0$ and $\int^{\infty}_0|\bar{\omega}_k(t)|dt < \infty$, $k=1,\cdots,N$. It is assumed that $\int^{\infty}_0\|\bar{\omega}_(t)\|_2dt < M$. Then for any $\epsilon > 0$, there exists $T_1$ such that $\int^{\infty}_{T_1}\|\bar{\omega}(t)\|_2dt < \epsilon$. Due to the convergence of $\Phi_{\bar{S}}(t,0)$, for any $\epsilon > 0$, there exists $T_2 > T_1$ such that for any $t_1\geq t_2\geq T_2$, $\|\Phi_{\bar{S}}(t_1,T_1) - \Phi_{\bar{S}}(t_2,T_1)\|_F < \epsilon$. Then the convergence of $\bar{S}(t)$ can be proved by the Cauchy's Convergence Theorem. That is
\begin{IEEEeqnarray}{lrl}\label{equ:tmp6}
& \|\bar{S}(t_1) - \bar{S}(t_2)\|_2 \leq & \|\Phi_{\bar{S}}(t_1,0)-\Phi_{\bar{S}}(t_2,0)\|_F\|\bar{S}(0)\|_2 + \int^{T_1}_0\|\Phi_{\bar{S}}(t_1,s)-\Phi_{\bar{S}}(t_2,s)\|_F\|\bar{\omega}(s)\|_2ds \IEEEnonumber\\
&&+\int^{t_1}_{T_1}\|\Phi_{\bar{S}}(t_1,s)-\Phi_{\bar{S}}(t_2,s)\|_F\|\bar{\omega}(s)\|_2ds + \int^{t_2}_{t_1}\|\Phi_{\bar{S}}(t_2,s)\|_F\|\bar{\omega}(s)\|_2ds\IEEEnonumber\\
& \leq & \|\bar{S}(0)\|_2N^2\epsilon + 2N^2M\epsilon + 2N^2\epsilon + N^2\epsilon.
\end{IEEEeqnarray}

Since $\lim_{t\to\infty}\bar{\omega}_k(t) = 0$, $k=1,\cdots,N$, by Lemma \ref{keylem}, it follows that $\lim_{t\to\infty}(\bar{s}_i(t)-\bar{s}_j(t)) = 0$ if $\mathcal{G}_{\sigma(t)}$ is uniformly jointly quasi-strongly connected. This together with the convergence of $\bar{S}(t)$ lead to that there exists $\bar{s}^*$ such that $\lim_{t\to\infty}\bar{s}_i(t) = \bar{s}^*$, $i=1,\cdots,N$. Because $A+K_3C$ is \emph{Hurwitz}, $\lim_{t\to\infty}K_2x_i(t) = \lim_{t\to\infty}K_2s_i(t) + \lim_{t\to\infty}K_2\exp((A+K_3C)t)(x_i(0)-s_i(0)) = \lim_{t\to\infty}\bar{s}_i(t) = \bar{s}^*$, $\forall i=1,\cdots,N$.
If all roots of (\ref{equ:characteristicequation}) are in the open left-half complex plane, by Lemma \ref{lem1}, it can be obtained that $\lim_{t\to\infty}x_i(t)  = (\bar{s}^*/a_1,0,\cdots,0)$, $i=1,\cdots,N$, which closes the proof of this theorem.
\end{proof}

\subsection{General Linear Time-Invariant Multi-Agent Systems}\label{sub:generallinear}
Consider the multi-agent system composed of the following general continuous-time linear time-invariant dynamical agent
\begin{equation}\label{equ:generallinearity}
\dot{x}_i(t) = A_gx_i(t) + B_gu_i(t),\quad y_i(t) = C_gx_i(t), \quad i=1,\cdots,N,
\end{equation}
where $x_i(t) \in \mathbb{R}^m$, $A_g \in \mathbb{R}^{m\times m}$, $B_g \in \mathbb{R}^{m}$, and $C_g \in \mathbb{R}^{m}$ ($C_g$ is a row vector).

If the pair $(A_g,B_g)$ is controllable and the pair $(A_g,C_g)$ is detectable, by Luenberger controllable canonical, there exists a matrix $T_g \in \mathbb{R}^{m\times m}$ such that
\begin{equation}
\bar{A}_g = T_gA_gT^{-1}_g = \begin{pmatrix}
0 & 1 & 0 & \cdots & 0\\
0 & 0 & 1 & \cdots & 0\\
\vdots & \vdots & \vdots & \ddots & \vdots\\
0 & 0 & 0 & \cdots & 1\\
a^1_g & a^2_g & a^3_g & \cdots & a^m_g
\end{pmatrix},\; \bar{B}_g = T_gB_g = \begin{pmatrix} 0\\0\\\vdots\\0\\1
\end{pmatrix}.
\end{equation}
Let $a_g = (a^1_g,a^2_g,\cdots,a^m_g)^T$. Then, by the similar proof of Theorem \ref{thm:main2}, the following protocol can solve the consensus problem of general linear multi-agent systems (\ref{equ:generallinearity}) if all roots of (\ref{equ:characteristicequation}) are in the open left-half complex plane and the communication topology $\mathcal{G}_{\sigma(t)}$ is uniformly jointly quasi-strongly connected.
\begin{IEEEeqnarray}{lll}\label{equ:tmp7}
&u_i(t) = -a^T_gT^{-1}_gs_i(t) + K_1T^{-1}_gs_i(t) - \sum_{j\in \mathcal{N}_i(\sigma(t))}\alpha^{ij}_{\sigma(t)}K_2T^{-1}_g(s_i(t) - s_j(t)), \IEEEyessubnumber  \\
& \dot{s}_i(t) = (A_g+K_3C_g)s_i(t) + B_gu_i(t) - K_3y_i(t), \IEEEyessubnumber
\end{IEEEeqnarray}
where $K_1$ and $K_2$ are defined in (\ref{equ:protocol}), and $K_3$ satisfies that $A_g+K_3C_g$ is \emph{Hurwitz}.

\section{Extensions to Discrete-Time High-Order Integral Multi-Agent Systems}\label{discrete}

The results obtained in the continuous-time domain can also be extended to the discrete-time case. Assume that the $i$th agent is described by the following discrete-time high-order model
\begin{IEEEeqnarray}{lll}\label{equ:discretegenerallinearity}
&{x}_i[k+1] = Ax_i[k] + Bu_i[k],\IEEEyessubnumber\\
&y_i[k] = Cx_i[k],\IEEEyessubnumber
\end{IEEEeqnarray}
where $x_i[k] = (x_{i,1}[k],\cdots,x_{i,m}[k])^T \in \mathbb{R}^m$, $u_i[k] \in \mathbb{R}$ and $y_i[k] \in \mathbb{R}$ are the agent's state, input and output at the $k$th step, respectively. The system matrices $A$, $B$ and $C$ are defined same as the ones in \eqref{equ:agentdynamicscompact} and \eqref{equ:output}.

The dynamic-output feedback based consensus protocol is proposed as follows
\begin{IEEEeqnarray}{ll}\label{equ:discreteobserverprotocol}
& u_i[k] = K_4z_i[k] - \frac{1}{1+\sum_{j\in \mathcal{N}_i[\sigma[k]]}\alpha^{ij}_{\sigma[k]}}\sum_{j\in \mathcal{N}_i[\sigma[k]]}\alpha^{ij}_{\sigma[k]}K_5(z_i[k] - z_j[k]), \IEEEyessubnumber \label{equ:discreteobserverprotocol1}\\
& {z}_i[k+1] = (A+K_6C)z_i[k] + Bu_i[k] - K_6y_i[k], \IEEEyessubnumber \label{equ:discreteobserverprotocol2}
\end{IEEEeqnarray}
where $\sigma[k]$ denotes the index of the graph representing the agent's communication topology at the $k$th step; $\mathcal{G}_{\sigma[k]}$, $\mathcal{N}_i[\sigma[k]]$ and $\alpha^{ij}_{\sigma[k]}$ have the similar meanings as the ones defined in the continuous-time case; $K_4 = (b_1,b_2-b_1,\cdots,b_{m-1}-b_{m-2}, 1 - b_{m-1})\in \mathbb{R}^{1\times m}$, $K_5 = (b_1,b_2,\cdots,b_{m-1},1) \in \mathbb{R}^{1\times m}$ and $K_6 \in \mathbb{R}^{m\times 1}$ is designed in such a way that all eigenvalues of $A+K_6C$ are in the unit circle.

By \eqref{equ:discreteobserverprotocol2}, it follows that
\begin{equation}
{z}_i[k+1]  - {x}_i[k+1] = (A+K_6C)({z}_i[k]  - {x}_i[k]),
\end{equation}
which results in
\begin{equation}\label{equ:discreteobserve}
{z}_i[k]  = (A+K_6C)^k({z}_i[0]  - {x}_i[0]) + {x}_i[k].
\end{equation}

Then substituting \eqref{equ:discreteobserverprotocol1} and \eqref{equ:discreteobserve} into \eqref{equ:discreteobserverprotocol2} obtains that
\begin{multline}\label{equ:discreteagentdynamicscompact}
Z[k+1] = (I_N\otimes(A+BK_4))Z[k]-(S[k]\otimes BK_5)Z[k]\\+ (I_N\otimes(K_6C(A+K_6C)^k))(Z[0]-X[0]),
\end{multline}
where $Z[k] = (z^T_1[k],\cdots,z^T_N[k])^T$, the $(i,i)$-entry of $S[k]$ is $({\sum_{j\in \mathcal{N}_i[\sigma[k]]}\alpha^{ij}_{\sigma[k]}})/({1+\sum_{j\in \mathcal{N}_i[\sigma[k]]}\alpha^{ij}_{\sigma[k]}})$ and the $(i,j)$-entry ($i\neq j$) of $S[k]$ is $-({\alpha^{ij}_{\sigma[k]}})/({1+\sum_{j\in \mathcal{N}_i[\sigma[k]]}\alpha^{ij}_{\sigma[k]}})$.

Similar with the continuous-time case, multiplying $I_N\otimes K_5$ at both sides of \eqref{equ:discreteagentdynamicscompact} obtains that
\begin{IEEEeqnarray}{lll}\label{equ:distransformed}
\bar{Z}[k+1] = (I_N - S[k])\bar{Z}[k] + \omega_d[k],
\end{IEEEeqnarray}
where $\omega_d[k] \equiv (I_N\otimes(K_5K_6C(A+K_6C)^k))(Z[0]-X[0])$, $\bar{Z}[k] \equiv (\bar{z}_1[k],\cdots,\bar{z}_N[k])^T$ and $\bar{z}_i[k] = K_5z_i[k]$ $(i=1,\cdots,N)$.

From \eqref{equ:distransformed}, it can be seen that the original discrete-time high-order integral multi-agent system has been transformed into a discrete-time first-order integral multi-agent system with a vanishing disturbance. Similar with the continuous-time case, in the following section, this reduced-order multi-agent system is proved to reach a consensus first, and then the consensus of the original multi-agent system can be ensured provided certain condition is applied to $(b_1,\cdots,b_{m-1})$.

\begin{lem}\label{lem2}
The discrete-time first-order integral multi-agent system defined by \eqref{equ:distransformed} can reach a consensus if
$\mathcal{G}_{\sigma[k]}$ is uniformly jointly quasi-strongly connected in the discrete-time sense.
($\mathcal{G}_{\sigma[k]}$ is called to be uniformly jointly quasi-strongly connected in the discrete-time sense if there exists a constant $M>0$ such that the union graph $(\mathcal{V}_\mathcal{G}, \bigcup_{i\in\{k,k+1,\cdots,k+M\}}\mathcal{E}_{\mathcal{G}_\sigma[i]})$ is quasi-strongly connected for any $k>0$.)
\end{lem}
\begin{proof}
Let the state transition matrix of \eqref{equ:distransformed} be
\begin{equation}\label{equ:distransitionmatrix}
\Phi_d(j,i) = \prod^{j-1}_{k=i}(I_N - S[k])\;\; \forall j>i, \quad \Phi_d(i,i) = I_N.
\end{equation}
It is easy to verify that $\forall k=0,1,\cdots$, $I_N-S[k]$ is a stochastic matrix with positive diagonal elements. Therefore, $\Phi_d(j,i)$ is also a stochastic matrix, which implies that $\|\Phi_d(j,i)\|_F < N^2$, $\forall j \geq i \geq 0$. Then the solution to \eqref{equ:distransformed} is
\begin{equation}\label{equ:distransitionmatrix}
\bar{Z}[k] = \Phi_d(k,0)\bar{Z}[0] + \sum^{k-1}_{i=0}\Phi_d(k,i)\omega_d[i].
\end{equation}
Next, two mathematical operators are defined, which play a key role in analyzing the convergence of $\bar{Z}[k]$. For a stochastic matrix $P=\{P_{ij}\} \in \mathbb{R}^{N\times N}$, define $\tau(P) = 0.5\max_{i,j}\sum^N_{s=1}|P_{is} - P_{js}|$. For a vector $p = (p_1,\cdots,p_N)^T \in \mathbb{R}^{N}$, define $\Delta p = \max_{i,j}|p_i-p_j|$. By \cite{Seneta81Book}, for any vector $r_1\in \mathbb{R}^N$ and stochastic matrix $P \in \mathbb{R}^{N\times N}$, if $r_2 = Pr_1$, then $\Delta r_2 \leq \tau(P)\Delta r_1$. Therefore it can be obtained that
\begin{equation}
\Delta\bar{Z}[k] \leq \tau(\Phi_d(k,0))\Delta\bar{Z}[0] + \sum^{k-1}_{i=0}\tau(\Phi_d(k,i))\Delta\omega_d[i].
\end{equation}

By Lemma 3.9 in \cite{Ren05TAC} and Theorem 3.1 in \cite{Wang09SIC}, the following properties hold for the state transition matrix $\Phi_d(\cdot,\cdot)$. That is: if the communication topology $\mathcal{G}_{\sigma[k]}$ is uniformly jointly quasi-strongly connected in the discrete-time sense, then
\begin{itemize}
\item there exists a vector $c \in \mathbb{R}^N$ satisfying $1^T_Nc = 1$ such that $\lim_{t\to\infty}\Phi_d(t,0) = 1_Nc^T$;
\item for any integers $k, h \geq 0$, there exist $C_M > 0$ and $0 < \lambda_{\Phi} < 1$ such that $\tau(\Phi_d(k+h,k)) \leq C_M\lambda^h_{\Phi}$.
\end{itemize}

Let the maximal magnitude eigenvalue of $A+K_6C$ be $\lambda_{max}$. Since all eigenvalues of $A+K_6C$ are in the unit circle, then $|\lambda_{max}| < 1$. And there must exist a constant $C_{e}$ such that $\Delta(\omega_d[k]) = \Delta((I_N\otimes(K_5K_6C(A+K_6C)^k))(Z[0]-X[0])) < C_e|\lambda_{max}|^k$. Therefore,
\begin{IEEEeqnarray}{lll}
\sum^{k-1}_{i=0}\tau(\Phi_d(k,i))\Delta\omega_d[i] \leq \sum^{k-1}_{i=0}C_MC_e\lambda^{k-i}_{\Phi}|\lambda_{max}|^i \leq (k-1)C_MC_e|\lambda_{min}|^{k},
\end{IEEEeqnarray}
where $|\lambda_{min}| = \max\{\lambda_{\Phi},|\lambda_{max}|\} < 1$. Then
\begin{equation}
\lim_{k\to\infty}\sum^{k-1}_{i=0}\tau(\Phi_d(k,i))\Delta\omega_d[i] = 0,
\end{equation}
which results in that
\begin{equation}\label{equ:limitbarZ}
\lim_{k\to\infty}\Delta\bar{Z}[k] \leq \lim_{k\to\infty}\tau(\Phi_d(k,0))\Delta\bar{Z}[0] + \lim_{k\to\infty}\sum^{k-1}_{i=0}\tau(\Phi_d(k,i))\Delta\omega_d[i] = 0.
\end{equation}
By the definition of the operator $\Delta$, it follows that $\forall k\geq 0$, $\Delta\bar{Z}[k] \geq 0$. This together with \eqref{equ:limitbarZ} leads to
\begin{equation}\label{equ:distmp2}
\lim_{k\to\infty}(\bar{z}_i[k] - \bar{z}_j[k]) = 0, \quad \forall 1\leq i,j \leq N.
\end{equation}

Similar with Theorem \ref{thm:main2}, the convergence of $\bar{Z}[k]$ is proved in the following part. From one side, since all eigenvalues of $A+K_3C$ are in the unit circle, it is easy to prove that there exists a constant $M_d$ such that $\sum^{\infty}_{i=0}\|\omega_d[i]\|_2 < M_d$. Therefore, $\forall \epsilon > 0$, there exists $M_1 > 0$ such that $\sum^{\infty}_{i=K_1}\|\omega_d[i]\|_2 < \epsilon$. From the other side, because of the convergence of $\Phi_d(k,0)$, $\forall \epsilon > 0$, there exists $M_2 > M_1$ such that for any $k_1\geq k_2\geq M_2$, $\|\Phi_{d}(k_1,M_1) - \Phi_{d}(k_2,M_1)\|_F < \epsilon$. By the above two facts, the convergence of $\bar{Z}(t)$ can also be proved by the Cauchy's Convergence Theorem. That is
\begin{IEEEeqnarray}{lrl}\label{equ:distmp1}
& \|\bar{Z}[k_1] - \bar{Z}[k_2]\|_2 \leq & \|\Phi_{d}(k_1,0)-\Phi_{d}(k_2,0)\|_F\|\bar{Z}[0]\|_2 + \sum^{M_1-1}_{s=0}\|\Phi_{d}(k_1,s)-\Phi_{d}(k_2,s)\|_F\|{\omega}_d[s]\|_2 \IEEEnonumber\\
&&+\sum^{k_1-1}_{s=M_1}\|\Phi_{d}(k_1,s)-\Phi_{d}(k_2,s)\|_F\|{\omega}_d[s]\|_2 + \sum^{k_2}_{s=k_1}\|\Phi_{d}(k_2,s)\|_F\|{\omega}_d[s]\|_2\IEEEnonumber\\
& \leq & 2\|\bar{Z}[0]\|_2N^2\epsilon + 2N^2M_d\epsilon + 2N^2\epsilon + N^2\epsilon.
\end{IEEEeqnarray}
From \eqref{equ:distmp2} and \eqref{equ:distmp1}, it can be obtained that there exists $\bar{z}^*$ such that $\lim_{k\to\infty}|\bar{z}_i - \bar{z}^*| = 0$, $i=1,\cdots,N$.
\end{proof}

The next Lemma bridges the discrete-time first-order integral multi-agent system and the original high-order multi-agent system, which plays a same role as its continuous-time counterpart (Lemma \ref{lem1}).
\begin{lem}\label{lem:discrete}
Consider the following non-homogeneous linear difference equation
\begin{equation}\label{equ:disdifeqn}
r[k+m-1] + b_{m-1}r[k+m-2] + \cdots + b_2r[k+1] + b_1r[k] = f[k],
\end{equation}
where $r[k] \in \mathbb{R}$ and $f[k]$ is a convergent sequence satisfying $\lim_{k\to\infty}f[k] = f^*$. Let the characteristic equation associated with (\ref{equ:lemdifeqn}) be
\begin{equation}\label{equ:discharacteristicequation}
s^{m-1} + b_{m-1}s^{m-2} + \cdots b_2s + b_1 = 0.
\end{equation}
Then,
for any initial state $(r[0],r[1],\cdots,r[m-2])$, $r[k]$ is convergent to $f^*/(b_1+\cdots+b_{m-1}+1)$ if all roots of (\ref{equ:discharacteristicequation}) are in the unit circle.
\end{lem}
\begin{proof}
See the proof in the Appendix.
\end{proof}

\begin{thm}
Under the proposed consensus protocol defined by \eqref{equ:discreteobserverprotocol}, the discrete-time high-order multi-agent system defined by \eqref{equ:discretegenerallinearity} can reach a consensus if the communication topology $\mathcal{G}_{\sigma[k]}$ is uniformly jointly quasi-strongly connected in the discrete-time sense, all eigenvalues of $A+K_6C$ are in the unit circle and the control parameters $(b_1,\cdots,b_{m-1})$ in \eqref{equ:discreteobserverprotocol} are designed in such a way that all roots of \eqref{equ:discharacteristicequation} are in the unit circle.
\end{thm}
\begin{proof}
If $\mathcal{G}_{\sigma[k]}$ is uniformly jointly quasi-strongly connected in the discrete-time sense, by Lemma \ref{lem2}, there exists $z^* \in \mathbb{R}$ such that $\lim_{k\to\infty}{\bar{Z}}[k] = z^*1_N$. Then by \eqref{equ:discreteobserve}, $\lim_{k\to\infty}\bar{Z}[k] = \lim_{k\to\infty}(I_N\otimes K_5){Z}[k] = \lim_{k\to\infty}(I_N\otimes K_5)((I_N\otimes(K_6C(A+K_6C)^k))(Z[0]-X[0]) + X[k]) = \lim_{k\to\infty}(I_N\otimes K_5)X[k]$ since all eigenvalues of $A+K_6C$ are in the unit circle. Therefore, $\lim_{k\to\infty}K_5x_i[k] = z^*$ which implies that
\begin{IEEEeqnarray}{ll}
&x_{i,1}[k+m-1]+b_{m-1}x_{i,1}[k+m-2]+\cdots + b_{2}x_{i,1}[k+1] + b_1x_{i,1}[k] = f^i_d[k], \IEEEnonumber\\
&\lim_{k\to\infty}f^i_d[k] = z^*,\quad i=1,\cdots,N.
\end{IEEEeqnarray}

Since all roots of \eqref{equ:discharacteristicequation} are in the unit circle, it follows by Lemma \ref{lem:discrete} that $\lim_{k\to\infty}x_{i,1}[k] = z^*$. Due to the relationship between $x_{i,1}[k],\cdots,x_{i,m}[k]$, it can be proved that $\lim_{k\to\infty}x_{i}[k] = z^*1_N$, $i=1,\cdots,N$, which closes the proof of this Theorem.
\end{proof}

\section{Simulation Examples}\label{Simulation}
In this section, three simulation examples are provided to validate the correctness of theoretical analysis.

\subsection{Example One}
Consider a network of five identical single-link flexible-joint robots. A sketch of this robot is shown in Fig. \ref{fig:robotsketch}. According to \cite{Slotine91Book}, the $i$th robot is modeled by the following dynamics
\begin{equation}\label{equ:robotic}
\begin{cases}
& I\ddot{q}_{i,1}(t) + MgL\sin q_{i,1}(t) + k(q_{i,1}(t)-q_{i,2}(t)) = 0, \\
& J\ddot{q}_{i,2}(t) - k(q_{i,1}(t)-q_{i,2}(t)) = \tau_i(t),
\end{cases}
\end{equation}
where $q_{i,1}(t)$ denotes the angle of the robotic link; $q_{i,2}(t)$ represents the angle of actuator; $I$ and $J$ denote the inertia of the actuator and robotic link, respectively; $M$ is the mass of the robotic link; $L$ represents the length between the mass center and the joint mounting point; $k$ denotes the spring's torsion coefficient; $\tau_i(t)$ is the torque input. The communication topology switches between the following three cases (shown in Fig. \ref{fig:topology}). $\sigma(t) = \mod(10t,3)+1$. Then $\forall t \geq 0$, the union graph $\mathcal{G}_{[t,t+0.3)}$ is quasi-strongly connected which means that $\mathcal{G}_{\sigma(t)}$ is uniformly jointly quasi-strongly connected. And $\alpha^{21}_{\mathcal{G}_1} = \alpha^{41}_{\mathcal{G}_1} = \alpha^{53}_{\mathcal{G}_1} = 1$; $\alpha^{15}_{\mathcal{G}_2} = \alpha^{23}_{\mathcal{G}_2} = \alpha^{24}_{\mathcal{G}_2} = 1$; $\alpha^{32}_{\mathcal{G}_3} = \alpha^{34}_{\mathcal{G}_3} =\alpha^{45}_{\mathcal{G}_3} = 1$; and all other $ \alpha^{ij}_{\mathcal{G}_k} = 0$, $i,j=1,\cdots,N, k=1,2,3$.  The control objective is to design a distributed protocol for each robot to drive all robots' joints to a same angle, i.e., there exists $q^*_1$ and $q^*_2$ such that $\lim_{t\to\infty}q_{i,1}(t) = q^*_1$ and $\lim_{t\to\infty}q_{i,2}(t) = q^*_2$, $i=1,\cdots,5$.

By the feedback linearization technique, let
\begin{IEEEeqnarray}{lll}\label{equ:robotdynamics}
& x_{i,1}(t) = q_{i,1}(t), \IEEEnonumber\\
& x_{i,2}(t) = \dot{q}_{i,1}(t), \IEEEnonumber\\
& x_{i,3}(t) = -MgL\sin(q_{i,1}(t))/I - k(q_{i,1}(t)- q_{i,2}(t))/I,\IEEEnonumber\\
& x_{i,4}(t) = -MgL\dot{q}_{i,1}(t)\cos(q_{i,1}(t))/I - k(\dot{q}_{i,1}(t) - \dot{q}_{i,2}(t))/I,
\end{IEEEeqnarray}
and
\begin{multline}\label{equ:robotcontrol}
\tau_i(t) = \frac{IJ}{k}\bigg(u_i(t) - \frac{MgL}{I}\sin(q_{i,1}(t))( \dot{q}^2_{i,1}(t) + \frac{MgL}{I}\cos(q_{i,1}(t)) + \frac{k}{I})\\+\frac{k}{I}(q_{i,1}(t)- q_{i,2}(t))(\frac{k}{I} + \frac{k}{J} + \frac{MgL}{I}\cos(q_{i,1}(t)))\bigg).
\end{multline}
The nonlinear robotic dynamics defined by \eqref{equ:robotic} can be transformed into the following high-order integral dynamics
\begin{equation}
\dot{x}_{i,1}(t) = x_{i,2}(t),\;\;\dot{x}_{i,2}(t) = x_{i,3}(t),\;\;\dot{x}_{i,3}(t) = x_{i,4}(t),\;\;\dot{x}_{i,4}(t) = u_{i}(t).
\end{equation}
Therefore, the consensus protocol $u_i(t)$ in \eqref{equ:robotcontrol} can be designed according to \eqref{equ:protocol}. $K_1$ and $K_2$ in \eqref{equ:protocol} are set as $K_1 = (0,-1,-3,-3)$ and $K_2 = (1,3,3,1)$. It is easy to verify that under these control parameters, all roots of \eqref{equ:lemdifeqn} are in the open left-half complex plane. By Theorem \ref{thm:main}, the control objective can be achieved. In the simulation, the parameters in \eqref{equ:robotic} are set as follows: $M = 1.5$kg; $g=9.8$m/s$^2$; $J=3.2$kg$\cdot$m$^2$;  $I=1$kg$\cdot$m$^2$; $L=0.8$m; $k=2.5$N/deg. The initial joint configuration is set as $q_{1,1}(0) = 2.5$deg, $q_{1,2}(0) = 1.5$deg, $q_{2,1}(0) = 1.9$deg, $q_{2,2}(0) = 3.14$deg, $q_{3,1}(0) = -2.4$deg, $q_{3,2}(0) = -2.6$deg, $q_{4,1}(0) = 1.57$deg, $q_{4,2}(0) = -1.5$deg, $q_{5,1}(0) = -3.14$deg, $q_{5,2}(0) = 0$deg and $\dot{q}_{i,1}(0) = \dot{q}_{i,2}(0) = 0$deg/s. The simulation results are shown in Fig. \ref{fig:joint1}. It can be seen that $\lim_{t\to\infty}q_{i,1}(t) = -3.255$deg and $\lim_{t\to\infty}q_{i,2}(t) =  -2.723$deg, which means that the consensus problem is solved. Therefore, the correctness of Theorem \ref{thm:main} is illustrated by this example. In addition, this example also shows that the consensus problem of some nonlinear multi-agent systems can be solved by combining the proposed protocol and the feedback linearization approach.

\subsection{Example Two}
This example studies the attitude consensus problem of a group of four aircrafts. The schematic diagram of the aircraft is shown in Fig. \ref{fig:aircraft}. A simplified dynamical model of the aircraft vertical motion is introduced in \cite{Slotine91Book}. That is
\begin{IEEEeqnarray}{ll}\label{equ:aircraft}
& J\ddot{\alpha}_i + b\dot{\alpha}_i + (C_{ZE}l + C_{ZW}d)\alpha_i = C_{ZE}lE_i, \IEEEyessubnumber\\
& m\ddot{h}_i = (C_{ZE} + C_{ZW})\alpha_i - C_{ZE}E_i, \IEEEyessubnumber
\end{IEEEeqnarray}
where $h_i$ denotes the aircraft's attitude; $E_i$ denotes the elevator rotation angle (control input); $\alpha_i$ denotes the rotation angle of the aircraft about its mass center $C_G$; $b$ is the friction coefficient;
$L_W = C_{ZW}\alpha_i$ is the lifting force applied at the ``center of lift'' $C_L$; $L_E = C_{ZE}(E_i-\alpha_i)$ is the aerodynamic force on the elevator; $m$ is the aircraft's mass and $J$ is its moment of inertia about $C_G$.
In \cite{Slotine91Book}, these parameters are set as $J =1$, $m=1$, $b=4$, $C_{ZE} = 1$, $C_{ZW} = 5$, $l=3$ and $d=0.2$. Let $x_i = (x_{i,1},x_{i,2},x_{i,3},x_{i,4})^T = (\alpha_i,\dot{\alpha}_i,h_i,\dot{h}_i)^T$,
then equation \eqref{equ:aircraft} can be written in the form of \eqref{equ:generallinearity} with
\begin{displaymath}
A_g = \begin{pmatrix}
0 & 1 & 0 & 0\\
-4 & -4 & 0 & 0\\
0 & 0 & 0& 1\\
6 & 0 & 0 & 0
\end{pmatrix},\quad B_g = \begin{pmatrix}0\\3\\0\\-1\end{pmatrix}.
\end{displaymath}
In this example, it is assumed that only the aircraft's attitude $h_i$ can be measured, which means that $y_i = C_gx_i = h_i$ with $C_g = (0,0,1,0)$ in \eqref{equ:generallinearity}. It is easy to verify that under these parameters, the attitude dynamics of aircraft is controllable and observable.   The switch signal $\sigma(t)$ of the communication topology is $\sigma(t) = \mod(10t,4)+1$. And $\alpha^{ij}_{\mathcal{G}_k} = 1$ if $e_{ji} \in \mathcal{G}_k$; otherwise $\alpha^{ij}_{\mathcal{G}_k} = 0$, $i,j \in\{1,\cdots, N\}$, $k\in\{1,2,3,4\}$. And it is easy to see that $\mathcal{G}_{\sigma(t)}$ is uniformly jointly quasi-strongly connected. Therefore, by the analysis in Section \ref{sub:generallinear}, the attitude consensus can be reached if the parameters in \eqref{equ:tmp7} are selected as follows
\begin{IEEEeqnarray}{ll}
&  K_1 = (0,-1,-3,-3), K_2 = (1,3,3,1), K_3 = (1/3,-2/3,-6,-7)^T,\IEEEnonumber\\
& T^{-1}_g = \begin{pmatrix}
0 & 0 & 3 & 0\\
0 & 0 & 0 & 3\\
14 & -4 & -1 & 0\\
0 & 14 & -4 & -1
\end{pmatrix}. \IEEEnonumber
\end{IEEEeqnarray}
In the simulation, four aircrafts' initial states are $x_1(0) = (0,0,8000,0)^T$, $x_1(0) = (0,0,6500,0)^T$, $x_1(0) = (0,0,7000,0)^T$, $x_1(0) = (0,0,5000,0)^T$; and the initial estimated states are $s_1(0) = (0,0,10000,0)^T$, $s_1(0) = (0,0,7000,0)^T$, $s_1(0) = (0,0,6000,0)^T$, $s_1(0) = (0,0,4000,0)^T$. The trajectory profiles of four aircrafts' real states and estimated states are displayed in Fig. \ref{fig:attitude}, which implies that four aircrafts' attitudes reach a consensus.

\subsection{Example Three}
To validate the theoretical analysis in Section \ref{discrete}, the consensus problem of a group of four discrete-time high-order integral agents defined by \eqref{equ:discretegenerallinearity} is studied. The parameters in \eqref{equ:discretegenerallinearity} are set as $m=4$ and $C=(1,0,0,0)$. The switching communication topology is similar with the one in Example Two ($\sigma[k] = \mod(k,4)+1, k=0,1,\cdots$). The control gain vectors in \eqref{equ:discreteobserverprotocol} are set as $K_4 = (1/8, 5/8, 3/4, -1/2)$, $K_5 = (1/8, 3/4, 3/2, 1)$ and $K_6 = (-2,-3/2,-1/2,-1/16)^T$. Then, all roots of \eqref{equ:discharacteristicequation} are in the unit circle. The initial states of four agents are randomly selected, and the initial estimated states are set as $0_4$. The trajectory profiles of $x_{i,1}[k]$ ($i=1,2,3,4$) are given in Fig. \ref{fig:discreteagent}, which illustrates the effectiveness of the proposed protocol defined by \eqref{equ:discreteobserverprotocol}.

\section{Conclusion}\label{Conclusion}
This paper studies the consensus of high-order integral multi-agent systems under the switching directed topology. The state-feedback based protocol and the dynamic output-feedback based protocol are proposed to solve the consensus problem, respectively. It is noted that by the proposed approaches, the consensus of the high-order integral multi-agent system can be transformed into the consensus of the first-order integral multi-agent system if all roots of the polynomial (\ref{equ:characteristicequation}) are in the open left-half complex plane. By using results in the robust consensus, it is proved that under both consensus protocols, the connectivity condition on the switching directed topology is only ``uniformly jointly quasi-strongly connected'', which is much weaker than the existing conditions. Finally, it should be noted that one most interesting contribution of this paper is that under the proposed protocols, there exists certain ``equivalence'' between the first-order integral multi-agent system and the high-order integral multi-agent system. Most results in the first-order integral multi-agent system can therefore be generalized to the high-order integral case, which deserves more investigation in the future.

\section*{Appendix: Proof of Lemma \ref{lem1}}
\textbf{Proof:} (Sufficiency) Let the roots of (\ref{equ:characteristicequation}) be $\lambda_1, \lambda_2, \cdots, \lambda_{m-1}$. Without loss of generality, it is assumed that $\lambda_1 = \cdots = \lambda_{k_1}$, $\lambda_{k_1+1} = \cdots = \lambda_{k_1+k_2}$, $\cdots$, $\lambda_{\sum^{l-1}_{i=1}k_i+1} = \cdots = \lambda_{\sum^{l}_{i=1}k_i}$ $(l \in \mathbb{N}, 1 \leq l \leq m-1, k_i \in \mathbb{N}, 1 \leq k_i \leq m-1, \sum^l_{i=1}k_i = m-1)$. By the knowledge of differential equation, one particular solution of (\ref{equ:lemdifeqn}) can be written in the following form
\begin{equation}\label{equ:particularsolution}
r^p(t) = e^{\lambda_1 t}\int^t_0e^{(\lambda_2-\lambda_1)\tau_1}\cdots\int^{\tau_{m-3}}_0e^{(\lambda_{m-1}-\lambda_{m-2})\tau_{m-2}}\int^{\tau_{m-2}}_0e^{-\lambda_{m-1}\tau_{m-1}}f(\tau_{m-1})d\tau_{m-1}d\tau_{m-2}\cdots d\tau_1.
\end{equation}

After obtaining the particular solution $r^p(t)$, the general solution of (\ref{equ:lemdifeqn}) can be expressed as follows
\begin{equation}\label{equ:generalsolution}
r(t) = r^p(t) + \sum^l_{j=1}\sum^{k_j}_{i=1}C_{\sum^{j-1}_{p=1}k_p+i}R_{ji}(t), \quad R_{ji}(t) = t^{i-1}e^{t\lambda_{\sum^{j-1}_{p=1}k_p+1}},
\end{equation}
where $C_1,\cdots,C_{m-1}$ are coefficients to be determined by the initial state.

Next, let us study the limitation property of $e^{\lambda t}\int^t_0e^{-\lambda s}g(s)ds$ where $\lambda$ is a complex constant with negative real part and $g(t)$ is a continuous function satisfying $\lim_{t\to\infty}g(t) = g^*$. Since $\lim_{t\to\infty}g(t) = g^*$, for any $\epsilon > 0$, there exists $t_1>0$ such that $|g(t) - g^*| < \epsilon$, $\forall t > t_1$. In addition, there exists a constant $M < \infty$ such that $|g(t)| < M$, $\forall t \geq 0$. Then
\begin{IEEEeqnarray}{lll}\label{equ:tmp1}
& \left|e^{\lambda t}\int^t_0e^{-\lambda s}g(s)ds + \frac{g^*}{\lambda}\right| = \left|\int^{t_1}_0e^{\lambda (t-s)}g(s)ds\right| +  \left|\int^t_{t_1}e^{\lambda (t-s)}g(s)ds  + \frac{g^*}{\lambda}\right| \leq e^{\lambda t}M\frac{1-e^{-\lambda t_1}}{\lambda} + \IEEEnonumber\\
& \left|\int^t_{t_1}e^{\lambda (t-s)}(g(s) - g^*)ds\right| + \left|\int^t_{t_1}e^{\lambda (t-s)}g^*ds + \frac{g^*}{\lambda}\right|  \leq e^{\lambda t}M\frac{1-e^{-\lambda t_1}}{\lambda} +   \epsilon\left(-\frac{1}{\lambda} + \frac{e^{\lambda(t-t_1)}}{\lambda}\right) + \IEEEnonumber\\
& g^*\left(\frac{e^{\lambda(t-t_1)}}{\lambda}\right).
\end{IEEEeqnarray}

Since $\lambda$ has the negative real part, (\ref{equ:tmp1}) leads to
\begin{equation}\label{equ:tmp2}
\lim_{t\to\infty}e^{\lambda t}\int^t_0e^{-\lambda s}g(s)ds = - \frac{g^*}{\lambda}.
\end{equation}

If $\lambda_1,\cdots,\lambda_{m-1}$ are all in the open left-half complex plane and $\lim_{t\to\infty}f(t) = f^*$,  applying (\ref{equ:tmp2}) to (\ref{equ:particularsolution}) by $(m-1)$ times obtains that
\begin{equation}\label{equ:tmp3}
\lim_{t\to\infty} r^p(t) = \frac{f^*}{\prod^{m-1}_{i=1}(-\lambda_i)} = \frac{f^*}{a_1}.
\end{equation}

In addition, if $\lambda_1,\cdots,\lambda_{m-1}$ are all in the open left-half complex plane, then $R_{ji}(t)$ in (\ref{equ:generalsolution}) has the property that $\lim_{t\to\infty}R_{ji}(t) = 0$. Then it can be obtained that whatever the coefficients $C_1,\cdots,C_{m-1}$ are, $r(t)$ in (\ref{equ:generalsolution}) always satisfies that $\lim_{t\to\infty}r(t) = f^*/a_1$. Furthermore, $r^{(1)}(t) = d\left({r}^p(t) + \sum^l_{j=1}\sum^{k_j}_{i=1}C_{\sum^{j-1}_{p=1}k_p+i}{R}_{ji}(t)\right)/dt$. And
\begin{multline}
\frac{dr^p(t)}{dt} = \lambda_1 e^{\lambda_1t}\int^t_0e^{(\lambda_2-\lambda_1)\tau_1}\cdots\int^{\tau_{m-2}}_0e^{-\lambda_{m-1}\tau_{m-1}}f(\tau_{m-1})d\tau_{m-1}\cdots d\tau_1\\+ e^{\lambda_2t}\int^t_0e^{(\lambda_3-\lambda_2)\tau_1}\cdots\int^{\tau_{m-3}}_0e^{-\lambda_{m-1}\tau_{m-2}}f(\tau_{m-2})d\tau_{m-2}\cdots d\tau_1.
\end{multline}

Therefore, if $\lambda_1,\cdots,\lambda_{m-1}$ all have the negative real parts, then $\lim_{t\to\infty} dr^p(t)/dt = \frac{(-\lambda_1)f^*}{\prod^{m-1}_{i=1}(-\lambda_i)} - \frac{f^*}{\prod^{m-1}_{i=2}(-\lambda_i)} = 0$. And it can also be obtained that $\lim_{t\to\infty}d{R}_{ji}(t)/dt = 0$, which leads to that $\lim_{t\to\infty}r^{(1)}(t) = 0$. By the same analysis, it can be proved that $\lim_{t\to\infty}r^{(i)}(t) = 0$, $i=1,\cdots,m-2$. Hence $\lim_{t\to\infty}(r(t),r^{(1)}(t),\cdots,r^{(m-2)}(t)) = (f^*/a_1,0,\cdots,0)$.

(Necessity) Let $r_g(t) = (r(t),r^{(1)}(t),\cdots,r^{(m-2)}(t))^T$, then (\ref{equ:lemdifeqn}) can be rewritten in the following form

\begin{equation}\label{equ:lemdifeqncompact}
\dot{r}_g(t) = \begin{pmatrix}
0 & 1 &  \cdots & 0\\
\vdots & \vdots &  \ddots & \vdots\\
0 & 0 & \cdots & 1\\
-a_1 & -a_2 & \cdots & -a_{m-1}
\end{pmatrix}r_g(t) + \begin{pmatrix}0\\\vdots\\0\\1\end{pmatrix}f(t) \equiv A_rr_g(t) + B_rf(t).
\end{equation}

Let the state transition matrix of (\ref{equ:lemdifeqncompact}) be $\Phi_r(t,s) = \exp(A_r(t-s))$. Then $r_g(t)$ can be solved as $r_g(t) = \Phi_r(t,0)r_g(0) + \int^t_0\Phi_r(t,s)B_rf(s)ds$.
If some roots of (\ref{equ:characteristicequation}) are not in the open left-half complex plane, by the knowledge of linear system, there exists at least one initial state $r^1_g(0)$ such that $\Phi_r(t,0)r^1_g(0)$ is not convergent as $t$ goes to infinity. Assume that under this initial state $r^1_g(0)$, we can still find a continuous function $f_1(t)$ satisfying $\lim_{t\to\infty}f_1(t) = f^*_1$ such that $r_g(t) = \Phi_r(t,0)r_g(0) + \int^t_0\Phi_r(t,s)B_rf_1(s)ds$ is convergent to $r^*_g$. From the other side, let us consider another continuous function $f_2(t) = 2f_1(t)$. Then under the same initial state $r^1_g(0)$, $r_g(t) + \Phi_r(t,0)r^1_g(0) = 2\Phi_r(t,0)r^1_g(0) + \int^t_0\Phi_r(t,s)B_rf_2(s)ds$ is convergent to $2r^*_g$. Since $\Phi_r(t,0)r^1_g(0)$ is not convergent, $r_g(t)$ with $f(t) = f_2(t)$ in (\ref{equ:lemdifeqncompact}) is not convergent. This contradicts with that $r_g(t)$ is convergent for any initial state $r^1_g(0)$.$\blacksquare$

\section*{Appendix: Proof of Lemma \ref{lem2}}
\textbf{Proof:}  This proof is similar with the one of Lemma \ref{keylem}. Note that one particular solution to \eqref{equ:disdifeqn} is
\begin{equation}
r^p[k] = \mu^k_{m-1}\sum^{k-1}_{i_{m-1}=0}\mu^{-i_{m-1}-1}_{m-1}\mu^{i_{m-1}}_{m-2}\sum^{i_{m-1}-1}_{i_{m-2}=0}\mu^{-i_{m-2}-1}_{m-2}\cdots\mu^{i_3}_2\sum^{i_3-1}_{i_2=0}\mu^{-i_2-1}_{2}\mu^{i_2}_1\sum_{i_1=0}^{i_2-1}\mu^{-i_1-1}_1f[i_1],
\end{equation}
where $\{\mu_1,\cdots,\mu_{m-1}\}$ are the roots of \eqref{equ:discharacteristicequation}. Assume that $\mu_1 = \cdots = \mu_{h_1}$, $\mu_{h_1+1} = \cdots = \mu_{h_1+h_2}$, $\cdots$, $\mu_{\sum^{l-1}_{i=1}h_i+1} = \cdots = \mu_{\sum^{l}_{i=1}h_i}$ $(l \in \mathbb{N}, 1 \leq l \leq m-1, h_i \in \mathbb{N}, 1 \leq h_i \leq m-1, \sum^l_{i=1}h_i = m-1)$. Then the general solution $r[k]$ to \eqref{equ:disdifeqn} can be written as a combination of $r^p[k]$ and $D_{ji}[k] = \mathbb{C}^{i}_{k}\left(\mu_{\sum^{j-1}_{p=1}h_p+1}\right)^{k-i}$ $(1 \leq j \leq l; 1 \leq i \leq h_j)$.

From one hand, the following property holds for $r^p[k]$
\begin{equation}
\lim_{k\to\infty}\mu^k\sum^{k-1}_{i=0}\mu^{-i-1}f[i] = f^*/(1-\mu), \quad \forall |\mu| < 1.
\end{equation}
Since all roots of \eqref{equ:disdifeqn} are in the unit circle, it follows that
\begin{equation}
\lim_{k\to\infty}r^p[k] = \frac{f^*}{\prod^{m-1}_i(1-\mu_i)} = f^*/(b_{m-1} + \cdots + b_1 + 1).
\end{equation}

From the other hand, it is easy to prove that $\lim_{k\to\infty}D_{ji}[k] = 0$ $(1 \leq j \leq l; 1 \leq i \leq h_j)$ if all roots of \eqref{equ:discharacteristicequation} are in the unit circle.

Therefore, the general solution to \eqref{equ:disdifeqn} is convergent to $f^*/(b_{m-1} + \cdots + b_1 + 1)$ regardless of the initial state. $\blacksquare$
\newpage
\bibliographystyle{IEEEtran}
\bibliography{IEEETIEbib}
\newpage
\listoffigures
\newpage
\begin{figure}[t]
\centering
\includegraphics[scale = 0.25]{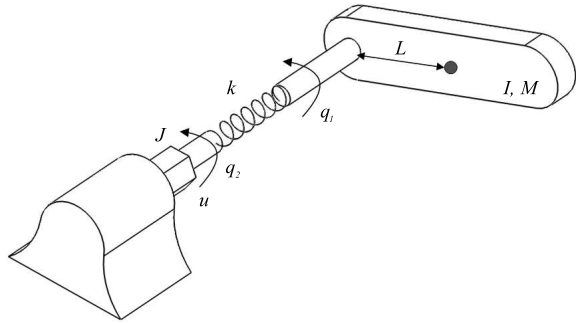}
  \caption{A sketch of a single-link flexible-joint robot.}\label{fig:robotsketch}
\end{figure}

\begin{figure}[t]
\centering \subfigure[]{
\includegraphics[scale = 0.34]{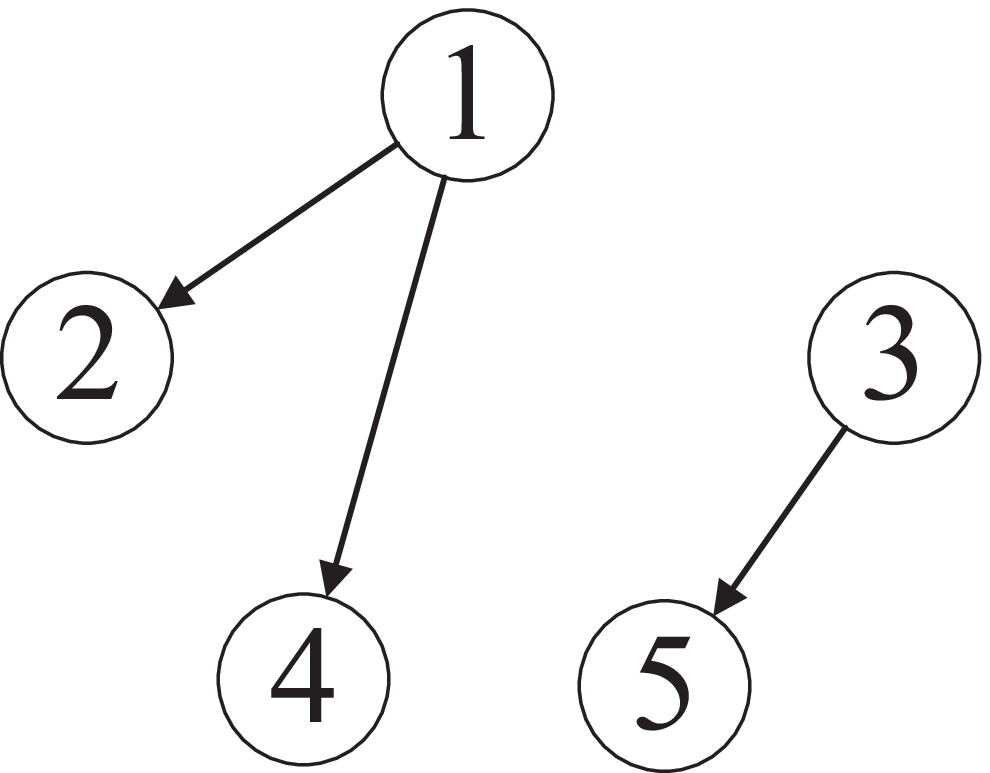}}
\hspace{0.1cm}\subfigure[]{
\includegraphics[scale = 0.35]{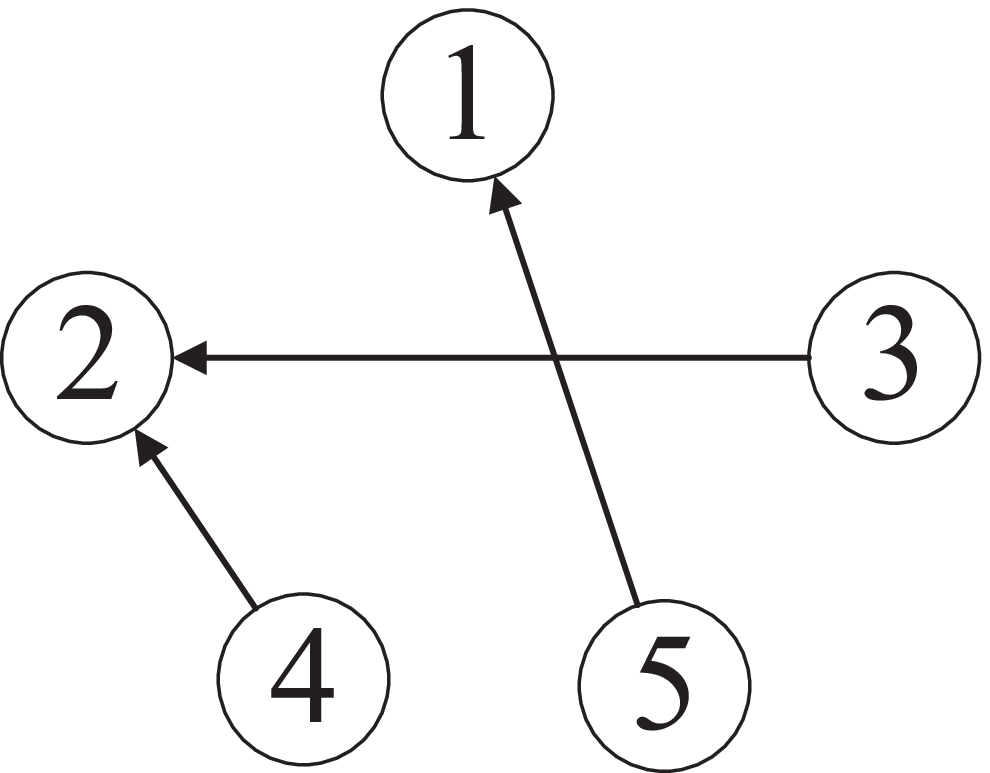}}
\hspace{0.1cm}\subfigure[]{
\includegraphics[scale = 0.35]{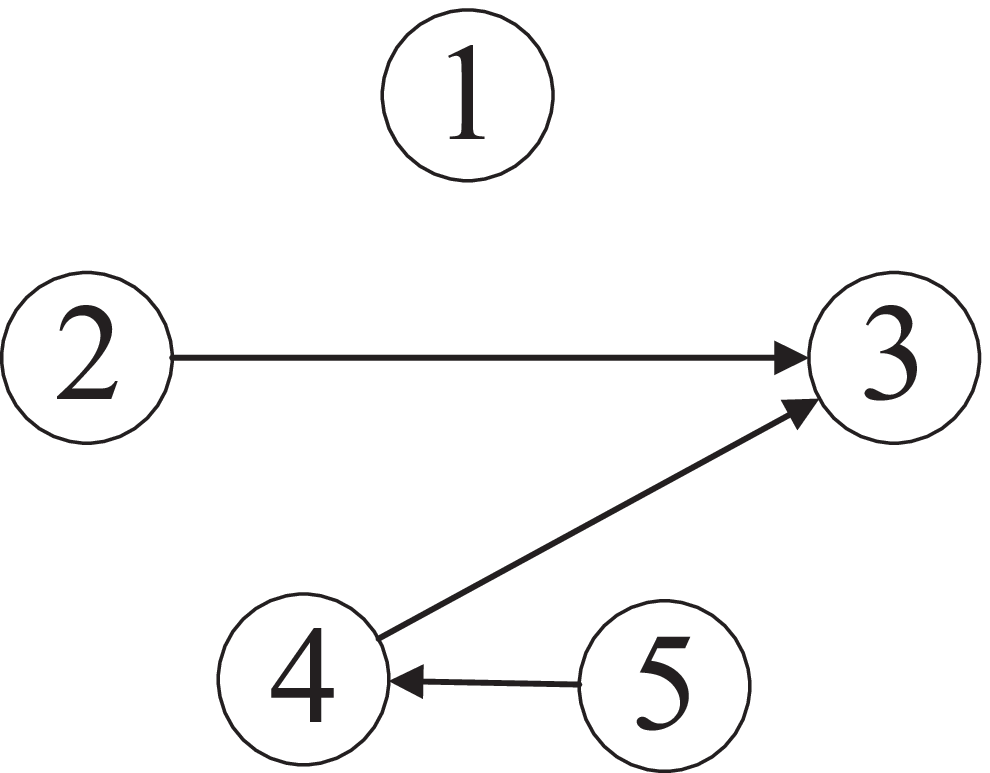}}
  \caption{Three possible communication topologies of the multi-agent system in Example One: (a) $\mathcal{G}_{1}$; (b) $\mathcal{G}_{2}$; and (c) $\mathcal{G}_{3}$.}\label{fig:topology}
\end{figure}

\begin{figure}[t]
\centering \subfigure[]{\tiny\psfrag{t}{Time (second)}\psfrag{a1}{$q_{1,1}(t)$}\psfrag{a2}{$q_{2,1}(t)$}\psfrag{a3}{$q_{3,1}(t)$}\psfrag{a4}{$q_{4,1}(t)$}\psfrag{a5}{$q_{5,1}(t)$}
\includegraphics[scale = 0.4]{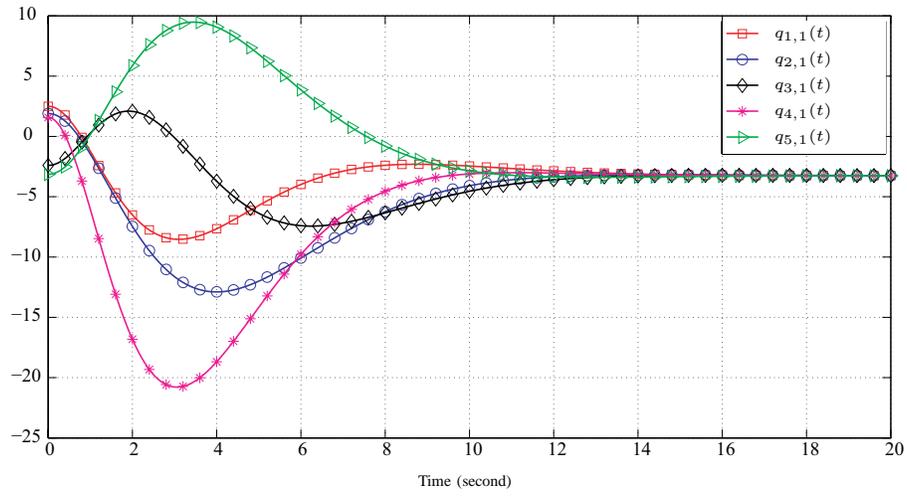}}
\subfigure[]{\tiny\psfrag{t}{Time (second)}\psfrag{a1}{$q_{1,2}(t)$}\psfrag{a2}{$q_{2,2}(t)$}\psfrag{a3}{$q_{3,2}(t)$}\psfrag{a4}{$q_{4,2}(t)$}\psfrag{a5}{$q_{5,2}(t)$}
\includegraphics[scale = 0.4]{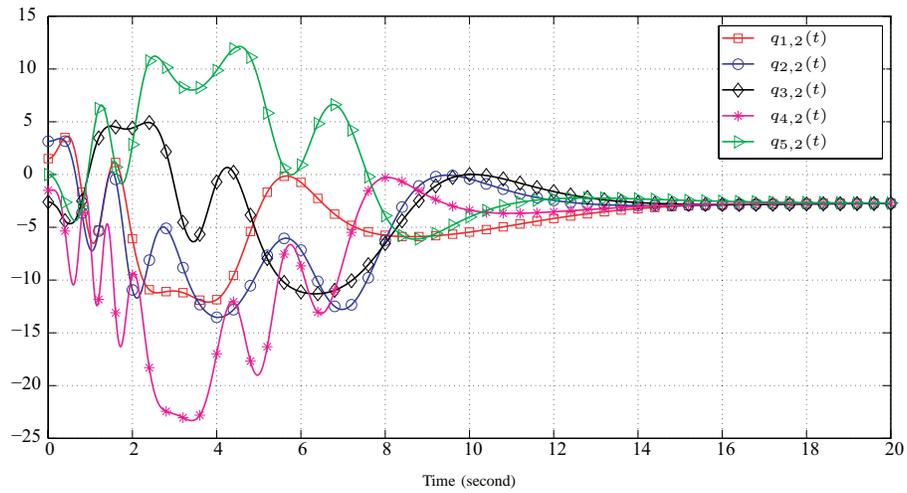}}
  \caption{The joint profiles of five single-link flexible-joint robots: (a) $q_{i,1}(t)$; (b) $q_{i,2}(t)$ $i=1,\cdots,5$.}\label{fig:joint1}
\end{figure}

\begin{figure}[t]
\centering\tiny\psfrag{alpha}{$\alpha$}\psfrag{Cg}{$C_G$}\psfrag{d}{$d$}\psfrag{l}{$l$}\psfrag{Cl}{$C_L$}\psfrag{Lw}{$L_W$}\psfrag{Le}{$L_E$}\psfrag{E}{$E$}
\includegraphics[scale = 0.5]{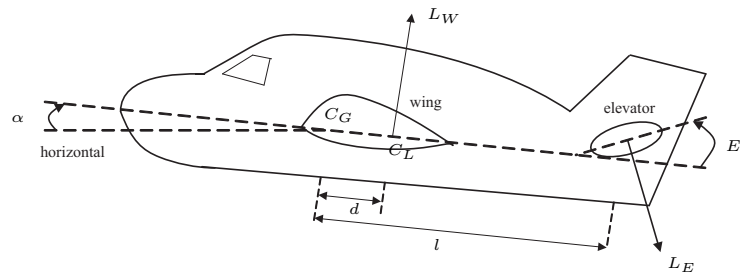}
  \caption{A schematic diagram of an aircraft.}\label{fig:aircraft}
\end{figure}

\begin{figure}[t]
\centering
\includegraphics[scale = 0.35]{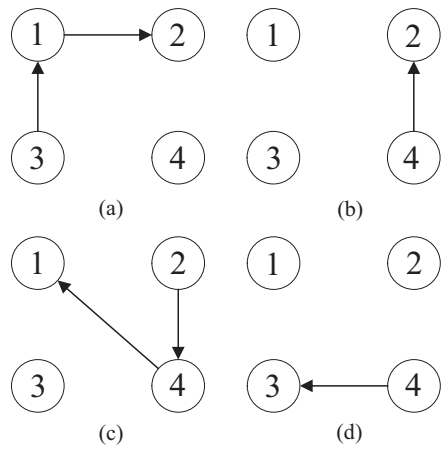}
  \caption{Four possible communication topologies of the multi-agent system in Examples Two and Three: (a) $\mathcal{G}_{1}$; (b) $\mathcal{G}_{2}$; (c) $\mathcal{G}_{3}$; (d) $\mathcal{G}_{4}$.}\label{fig:topology-ex2}
\end{figure}

\begin{figure}[t]
\centering \subfigure[]{\tiny\psfrag{t}{Time (second)}\psfrag{a1}{the real attitude of aircraft 1}\psfrag{a2}{the real attitude of aircraft 2}\psfrag{a3}{the real attitude of aircraft 3}\psfrag{a4}{the real attitude of aircraft 4}
\includegraphics[scale = 0.4]{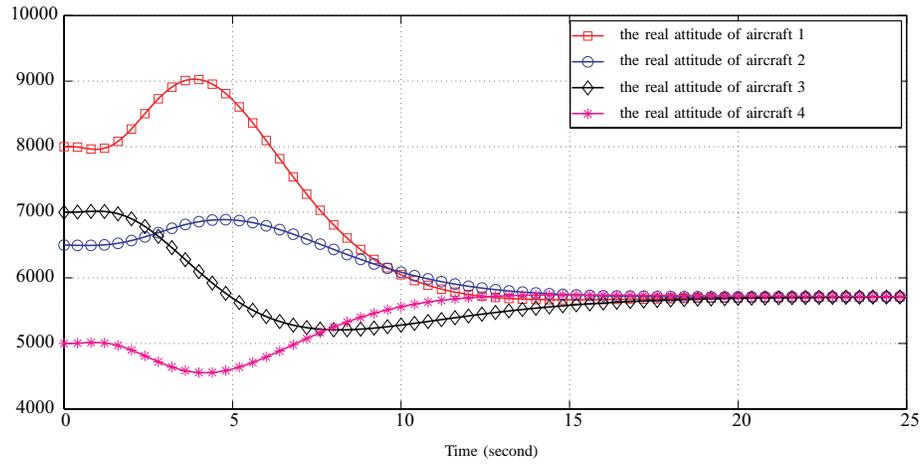}}
\subfigure[]{\tiny\psfrag{t}{Time (second)}\psfrag{a1}{the estimated attitude of aircraft 1}\psfrag{a2}{the estimated attitude of aircraft 2}\psfrag{a3}{the estimated attitude of aircraft 3}\psfrag{a4}{the estimated attitude of aircraft 4}
\includegraphics[scale = 0.4]{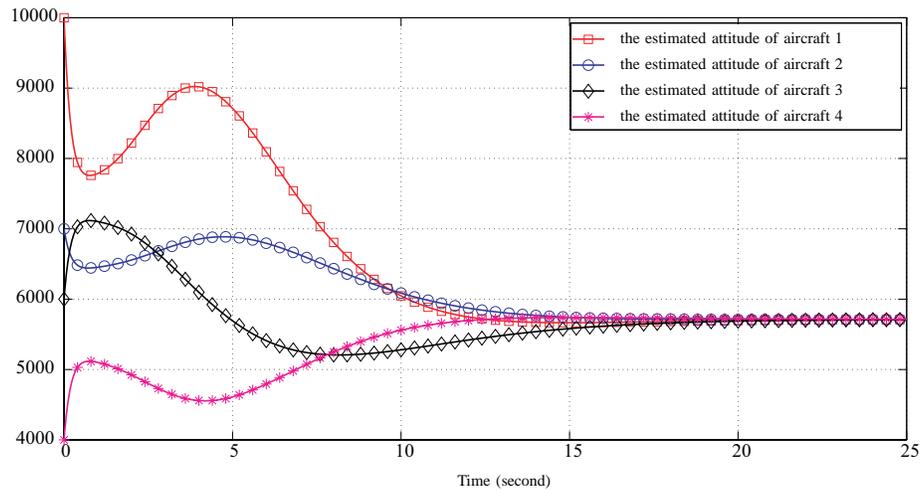}}
  \caption{The trajectory profiles of four aircrafts' real attitudes and estimated attitudes: (a) the profiles of the real attitudes; (b) the profiles of the estimated attitudes.}\label{fig:attitude}
\end{figure}

\begin{figure}[t]
\centering \subfigure[]{\tiny\psfrag{k}{Step (k)}\psfrag{a1}{$x_{1,1}[k]$}\psfrag{a2}{$x_{2,1}[k]$}\psfrag{a3}{$x_{3,1}[k]$}\psfrag{a4}{$x_{4,1}[k]$}
\includegraphics[scale = 0.4]{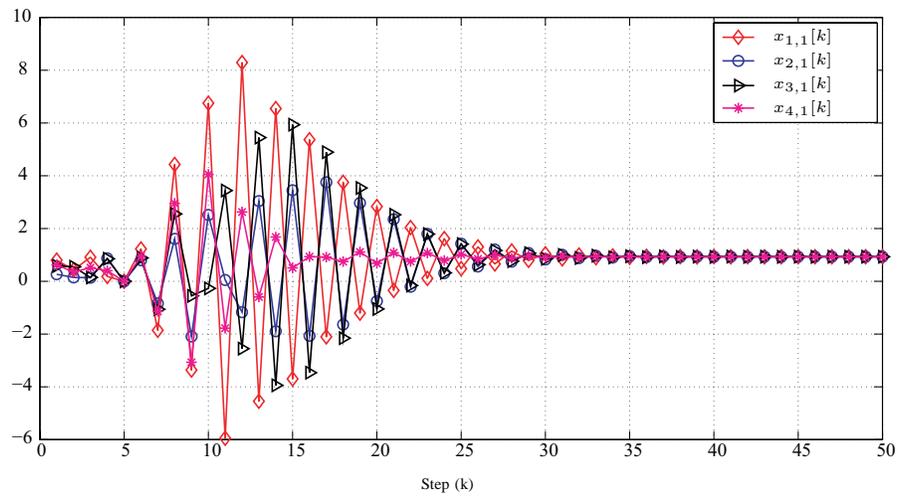}}
  \caption{The trajectory profiles of four agents' first-dimensional states.}\label{fig:discreteagent}
\end{figure}

\end{document}